\documentclass{article}
\usepackage[T1]{fontenc}
\usepackage[utf8]{inputenc}
\usepackage{cite}
\usepackage{amsmath,amssymb,amsfonts,amsthm}
\usepackage[a4paper, total={6in, 9in}]{geometry} %The parameter describe the area on
\usepackage{mathtools}
\usepackage{siunitx}
\usepackage{algorithm}
\usepackage{algpseudocode}
\usepackage{etoolbox}
\usepackage{graphicx}
\usepackage[caption=false,font=footnotesize]{subfig}
\usepackage[justification=centering]{caption}
\usepackage{bookmark}% bookmark in pdf file
\usepackage{hyperref}% references in pdf file
\usepackage{textcomp}
\usepackage{xcolor}
\usepackage{enumitem}

\makeatletter
\def\algbackskip{\hskip-\ALG@thistlm}
\newcommand*{\algrule}[1][\algorithmicindent]{%
	\makebox[#1][l]{%
		\hspace*{.2em}% <------------- This is where the rule starts from
		\vrule height .75\baselineskip depth .25\baselineskip
	}
}

\newcount\ALG@printindent@tempcnta
\def\ALG@printindent{%
	\ifnum \theALG@nested>0% is there anything to print
	\ifx\ALG@text\ALG@x@notext% is this an end group without any text?
	\else
	\unskip
	\ALG@printindent@tempcnta=1
	\loop
	\algrule[\csname ALG@ind@\the\ALG@printindent@tempcnta\endcsname]%
	\advance \ALG@printindent@tempcnta 1
	\ifnum \ALG@printindent@tempcnta<\numexpr\theALG@nested+1\relax
	\repeat
	\fi
	\fi
}
\patchcmd{\ALG@doentity}{\noindent\hskip\ALG@tlm}{\ALG@printindent}{}{\errmessage{failed to patch}}
\patchcmd{\ALG@doentity}{\item[]\nointerlineskip}{}{}{} % no spurious vertical space
\makeatother

\theoremstyle{plain}
\newtheorem{thm}{Theorem}

\newtheorem{cor}[thm]{Corollary}
\theoremstyle{definition}
\newtheorem{defn}[thm]{Definition}
\theoremstyle{remark}

\newcommand{\conv}[3]{#1 \otimes #2 \, (#3)}
\newcommand{\posPart}[1]{\left[#1\right]^+}
\newcommand{\IndicFun}[1]{\mathbf{1}_{\left\{ #1\right\}}}

\begin{document}

%\title{A New Network~Calculus Analysis of (Interleaved) Weighted~Round~Robin Using Bandwidth-Sharing~Policies}

\title{Improving Performance Bounds for Weighted~Round-Robin Schedulers under Constrained~Cross-Traffic}

\author{
	Vlad-Cristian Constantin, Paul Nikolaus, Jens Schmitt\\
	Distributed Computer Systems (DISCO) Lab\\
	TU Kaiserslautern\\
	\texttt{\{constantin,nikolaus,jschmitt\}@cs.uni-kl.de}
}
\date{}

\maketitle

\begin{abstract}
Weighted round robin (WRR) is a simple, efficient packet scheduler providing low latency and fairness by assigning flow weights that define the number of possible packets to be sent consecutively.
A variant of WRR that mitigates its tendency to increase burstiness, called interleaved weighted round robin (IWRR), has seen analytical treatment recently \cite{TLBB21}; a network calculus approach was used to obtain the best-possible strict service curve.
From a different perspective, WRR can also be interpreted as an emulation of an idealized fair scheduler known as generalized processor sharing (GPS).
Inspired by profound literature results on the performance analysis of GPS, we show that both, WRR and IWRR, belong to a larger class of fair schedulers called bandwidth-sharing policies.
We use this insight to derive new strict service curves for both schedulers that, under the additional assumption of constrained cross-traffic flows, can significantly improve the state-of-the-art results and lead to smaller delay bounds.
\end{abstract}

\section{Introduction}

For a long time, round-robin schedulers have been found appealing for their simplicity and corresponding efficient implementation as well as their inherent fairness \cite{Hahne91} (see \cite{Kleinrock64} for an early reference). 
%(Paper by Hahne in 1991, or even earlier like with Gallagher in 1986 which is in parallel to Katevenis from 1987 ...)
%a round-robin scheduler is already discussed in Kleinrock, 1964
Weighted round robin (WRR) is a frequently used scheduling algorithm in packet-switched networks as well as in real-time processing systems to provide a different resource allocation among flows (or tasks).
%It has a complexity of $\mathcal{O}(1)$ work per packet.
In its basic version, sometimes called plain WRR, we (conceptually) have a queue for each flow at a server and service is provided in rounds; in each round, a flow $ f_i $ receives the opportunity to send $ w_i $ packets consecutively.
The term weighted round robin was coined in \cite{KSC91} in the context of ATM (i.e., a network with constant packet sizes), where also some modifications, such as interleaving, were suggested.
WRR has been considered intensively in the literature on communication networks: e.g., investigating variants such as multi-class or multi-server WRR \cite{CM99, XJ04}; 
% Yuming2003-IEICE
or, in applications such as in the IEEE Standard 802.1Q \cite{IEEE802.1Q}, load balancing of cloud infrastructures \cite{LDK13, WC14}, or networks on chip (NoC) \cite{QLD09, HKR12}; and found usage in real-world equipment, e.g., in Ethernet switches \cite{HPEFlexNetwork}.
In contrast to another popular round-robin scheduler, deficit round robin (DRR) \cite{SV95}, WRR does not assume that the size of the head-of-the-line packet is known at each queue.
Therefore, it is also used in distributed queueing scenarios, for instance, as the uplink scheduler in an IEEE Standard 802.16 network \cite{CLME06}.

WRR has a tendency to increase the burstiness of flows due to several packets being sent back-to-back by the same flow per round. To mitigate this, interleaved weighted round robin (IWRR) introduces \emph{cycles} within a round to disperse the packet transmissions from the same flow.
This change is simple and the algorithm's complexity remains favorable (at $ \mathcal{O}(1) $ work per packet \cite{SV95}), though, it makes the mathematical performance analysis more challenging.
The performance analysis of IWRR has attracted some recent attention \cite{TLBB21}, using network calculus.

Network calculus \cite{Cruz91_1, Cruz91_2, Chang00, LBT01, BBLC18} is a versatile deterministic framework to derive worst-case per-flow performance guarantees.
To that end, deterministic constraints on arrivals and service are assumed.
These constraints are abstracted by so-called arrival and service curves, respectively.
While the former allows us to forego any distribution assumptions on inter-arrival times, the latter comes with the power of scheduling abstraction \cite{CS12-1}.
The so-called leftover service curve can then directly be used to calculate per-flow performance bounds.
Network calculus's level of modularization reduces the problem of improving performance bounds to finding a larger leftover service curve.

\begin{figure}[t]
	\centering
	\includegraphics[height = 0.2\textheight]{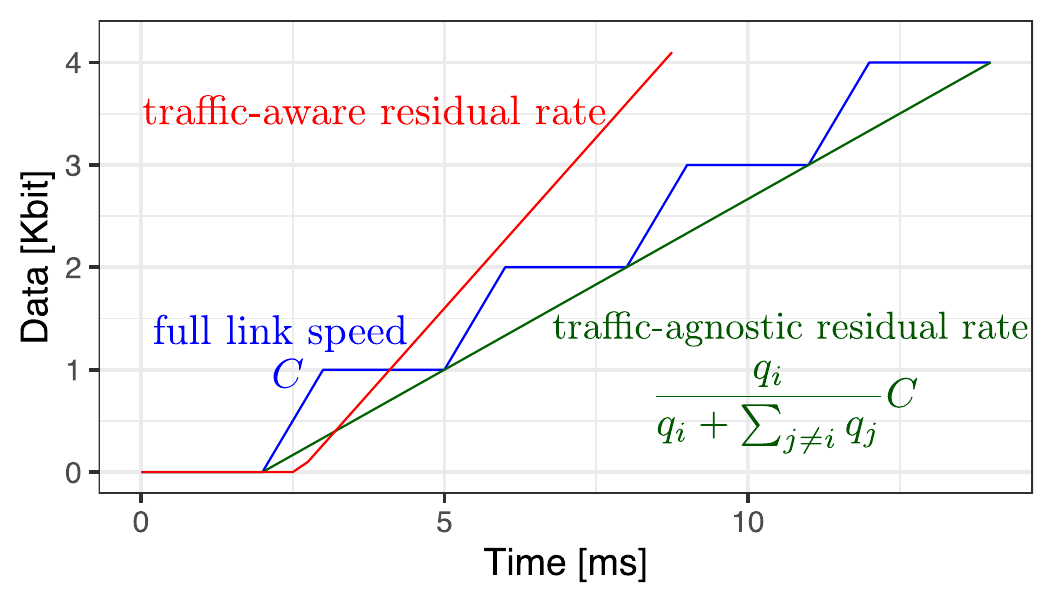}
	\caption{\centering State-of-the-art service curves for WRR and our service curve under constrained cross-traffic. \label{fig:linear-staircase}}
	%	\caption{Stair and shifted linear leftover service curve for unit-sized packets. \label{fig:linear-staircase}}
\end{figure}

Several different leftover service curves for WRR have been provided in the network calculus literature.
In \cite[p. 200]{BBLC18}, three service curves are derived for WRR.
The first one exploits knowledge on possible packet sequences by introducing so-called \emph{packet curves} \cite{BFG12}. 
Since we do not make this assumption, we only focus on the latter two.
%Let us assume for the moment that the server provides a constant rate $ C $ and unit-sized packet lengths.
Let us assume for the moment that the server provides a constant rate $ C. $ 
The least data per round for the flow of interest $ f_i $ is $ q_i \coloneqq w_i l_i^{\min} $ and the most data for a cross-flow $ f_{j \neq i} $ is $ q_j \coloneqq w_j l_j^{\max} $; here $ l_i^{\min} $ and $ l_j^{\max} $ denote the minimum and maximum packet size of the respective flow.
One leftover service curve is a stair function that closely models the packet scheduler's behavior of alternating between \emph{full link speed} (rate $ C $) and plateaus (rate 0) within a single round \cite{BD22}; the other leftover service curve has a shifted linear shape (with some initial latency) with \emph{traffic-agnostic residual rate} 
$ \frac{q_i}{q_i + \sum_{j \neq i} q_j} C $ 
% $ \frac{w_i}{\sum_{j=1}^{n} w_j} C $
that is just below the stair (illustrated in Figure~\ref{fig:linear-staircase}).
For IWRR, a leftover service curve for IWRR is derived that dominates all service curves obtained by the WRR analysis \cite{TLBB21}.

The aforementioned leftover service curves do not take into account constraints on the cross-traffic, even though it is a common case to have such constraints in application scenarios where worst-case performance guarantees are desired.
A key observation is that, given such constraints, all cross-flows will not remain backlogged for longer time intervals.
In fact, this observation has already been used to improve performance bounds for generalized process sharing (GPS), an idealized fair scheduler that achieves nearly perfect isolation and fairness \cite{GM90}.
It appeared in the seminal work on GPS by Parekh and Gallagher \cite{PG93} exploiting the \emph{feasible ordering} of flows. 
Later, this concept was generalized to \emph{feasible partitions} \cite{ZTK95}, larger classes of arrival curves and service curves \cite[pp. 68]{Chang00}, \cite{BL18}, \cite[pp. 172]{BBLC18}, and in a recent publication to a larger, practically more relevant class of fair schedulers, called \emph{bandwidth-sharing policies} \cite{Bouillard21}.
Weighted round robin, from a different perspective, can be interpreted as a GPS emulation.
As a consequence, we could benefit from the profound GPS results also in the WRR analysis.

%Contributions  ... Jens: may need some more work, refer to numerical evaluation results, think about the sentence regarding not considering the stair ...
In this work, under the assumption of constrained cross-flows, we provide new strict leftover service curves for WRR and IWRR using the network calculus framework.
Essentially, these are based on mathematical proofs that WRR and IWRR are bandwidth-sharing policies.
The new service curves can lead to significantly better delay bounds compared to the state of the art.
The reason for the improvement is that the new service curves' \emph{traffic-aware residual rate} is larger than their traffic-agnostic counterpart (see again Figure~\ref{fig:linear-staircase}; of course, details follow below).

The rest of the paper is structured as follows:
We provide the necessary background on network calculus in Section~\ref{sec:NC-background}.
In Section~\ref{sec:SOTA-BWSP}, we present the state of the art on WRR and IWRR as well as a class of fair schedulers called bandwidth-sharing policies.
In Section~\ref{sec:WRR-new-SC}, we show that WRR and IWRR, respectively, are bandwidth-sharing policies and derive new leftover service curves.
We provide numerical evaluations in Section~\ref{sec:numerical-evaluation}.
Section~\ref{sec:conclusion} concludes the paper.

\section{Network Calculus Background}
\label{sec:NC-background}

In this section, we present the necessary network calculus definitions and theorems as we use them throughout this paper.

An \emph{arrival process} (or input function) $A(t)$ of a flow $ f $ cumulatively counts the number of work units that arrive at a server in the interval $ [0,t) $.
We define it as an element of $ \mathcal{F}, $ the set of all wide-sense increasing functions:
\[
\hspace{-2mm} 
\mathcal{F} = \left\{ f:\mathbb{R}^+\rightarrow\mathbb{R}^+\cup\left\{ +\infty\right\} \mid\forall 0\leq s\leq t: 0\leq f(s)\leq f(t)\right\}.
\]
Moreover, we use the shorthand notation $ A(s,t) \coloneqq A(t) - A(s). $
Similarly, we denote its according \emph{departure process} by $ D(t) \in \mathcal{F}. $
We assume a system to be \emph{causal}, i.e., no data are created at the system: $ A(t) \geq D(t) $ for all $ t \geq 0 $ and write $ D(s,t) $ for $ D(t) - D(s). $
Furthermore, we assume all systems to be lossless.

\begin{defn}[Virtual Delay]
	\label{def:virtual-delay}
	The \emph{virtual delay} of data arriving at a server at time $t$ is the time until this data would be served, assuming FIFO order of service,
	\begin{equation}
		d(t) \coloneqq \inf\left\{ \tau\geq 0: A(t)\leq D(t+\tau)\right\}.
		\label{eq:virtual-delay}
	\end{equation}
\end{defn}

In order to provide worst-case performance guarantees, we need upper bounds on arrivals and lower bounds on the service, respectively.
We start off by defining arrival curves.

\begin{defn}[Arrival Curve]
	\label{def:arrival-curve}
	Given an increasing function $ \alpha \in \mathcal{F} $.
	We say that $ \alpha $ is an \emph{arrival curve} for an arrival process $ A $ if for all $ 0\leq s \leq t $
	\begin{equation*}
		A(t) - A(s) \leq \alpha(t-s). 
		\label{ineq:arrival-curve}
	\end{equation*}
\end{defn}

The most important example is the \emph{token bucket} arrival curve $ \gamma_{r,b}(t) = b+r\cdot t $ for $ t \geq 0. $
%if $ t \geq 0, $ and $ 0, $ else.
% Paul: note that we need to ensure concavity.
The parameter $ r $ denotes the rate and $ b $ the burst tolerance \cite[p. 7]{LBT01}.

For the service we need to introduce different notions.
First, we define service curves as they are needed to provide performance bounds.
Then, we introduce a stronger notion of so-called strict service curves, that we use for a per-flow analysis in a multi-flow system.
Afterwards, we define variable capacity nodes which are mainly used as a technical assumption.

\begin{defn}[Service Curve]
	\label{def:service-curve}
	Consider an arrival process $ A $ traversing a server and its according departure process $D$.
	The server offers a \emph{(minimum) service curve} $ \beta $ to $ A $ if $ \beta \in \mathcal{F} $ and for all $ t \geq 0 $
	\begin{equation*}
		D(t) \geq \conv{A}{\beta}{t} = \inf_{0\leq s\leq t} \left\{A(t-s) + \beta(s)\right\}.
		\label{ineq:service-curve}
	\end{equation*}
\end{defn}

We define a \emph{backlogged period} such that $ D(\tau) < A(\tau)  $ for all $ \tau \in (s,t]. $ 

\begin{defn}[Strict Service Curve]
	\label{def:strict-service-curve}
	A server is said to offer a \emph{strict service curve} $\beta \in \mathcal{F}$ to a flow if, during any backlogged period $ (s,t] $,
	\begin{equation*}
		D(s,t) \geq \beta(t-s).
		\label{ineq:strict-service-curve}
	\end{equation*}
\end{defn}

The most important example of service curves we employ is the \emph{rate-latency} service curve $ \beta_{R,T}(t) \coloneqq R\cdot \posPart{t-T}, $ where $ \posPart{x} \coloneqq \max\{x, 0\} $ denotes the positive part.

We define the function $ C(s,t) = C(t) - C(s) $ as the cumulative service process of a server for $ 0 \leq s \leq t $.
\begin{defn}[Variable Capacity Node]
	\label{def:vcn}
	A server is a \emph{variable capacity node} (VCN) if it offers $ \beta \in \mathcal{F} $ such that for \emph{any} $ 0\leq s \leq t $
	\[
	C(s,t) \geq \beta(t-s).
	\]
\end{defn}
%\begin{defn}[Variable Capacity Node]
%	\label{def:vcn}
%	A server is a \emph{variable capacity node} (VCN) if for $ 0\leq s \leq t $
%	\[
%	D(s,t) \begin{cases*}
	%		= C(s,t), & if $ (s,t] $ is a backlogged period,\\
	%		\leq C(s,t), & otherwise.
	%	\end{cases*}
%	\]
%\end{defn}
%Variable capacity nodes are basically just a technical assumption we need throughout this paper.
Typical examples such as constant-rate servers belong to this class.
%For the special case of a VCN, $ \beta $ is a strict service curve if 
%\[
%	C(s,t) \geq \beta (t-s), \quad \text{for all } 0 \leq s \leq t.
%\]
The assumption of a VCN is very mild, as it has been proven that it is equivalent to a strict service curve if the asymptotic growth rate of $ \beta $ is finite \cite[pp. 222]{BBLC18}.

Under these basic concepts, we are able to derive \emph{tight} performance bounds on the delay:
\begin{thm}[Delay Bound]
	\label{thm:delay-bound}
	Assume an arrival process $ A $ traversing a server.
	Further, let the arrivals be constrained by arrival curve $ \alpha $ and let the system offer a service curve $ \beta. $
	The virtual delay $d(t)$ satisfies for all $t$
	\begin{equation*}
		d(t)\leq \sup_{t\geq0}\left\{ \inf\left\{ d \geq 0 \mid \alpha(t)\leq \beta(t+d)\right\}\right\} \eqqcolon h(\alpha, \beta),
		\label{ineq:delay-bound}
	\end{equation*}
	where $ h(\alpha, \beta) $ is the \emph{horizontal deviation} between $ \alpha $ and $ \beta. $
\end{thm}
Tight means that we can create a sample path such that the delay is equal to its delay bound.

Using network calculus, we derive per-flow performance bounds.
Throughout the rest of this paper, if not stated otherwise, our flow of interest has the index $ i \in \left\lbrace 1, \dots, n \right\rbrace \eqqcolon \mathcal{N}  $
and we call the remaining $ n-1 $ flows cross-traffic.

\section{State-of-the-Art on (I)WRR and \\ Opening a Door for Improvement}
\label{sec:SOTA-BWSP}

In this section, we first explain the basic mechanics behind weighted round robin (WRR) and its interleaving variant, IWRR. 
Next, we present the state-of-the-art for the network calculus analysis of both, WRR and IWRR.
At the end of this section, we introduce a general class of fair schedulers called bandwidth-sharing policies and explain how we can leverage from this abstraction.

\subsection{Basics on (I)WRR}

\begin{algorithm}[t]
	\caption{Weighted Round Robin \cite[p. 200]{BBLC18}} \label{alg:WRR}
	\hspace*{\algorithmicindent} \textbf{Input} Integer weights $ w_1, \dots, w_n $
	\begin{algorithmic}[1]
		\While{True} \Comment{A round starts.}
		\For{$ i = 1 $ \textbf{to} $ n $}
		\State $ k \gets 1; $
		\While{\textbf{not} empty$ (i) $ \textbf{and} $ k \leq w_i $}
		\State \verb|send|(head($ i $));
		\State \verb|removeHead|($ i $);
		\State $ k \gets k + 1; $
		\EndWhile
		\EndFor
		\EndWhile \Comment{A round finishes.}
	\end{algorithmic}
\end{algorithm}

We start off with plain WRR:
Conceptually, the packets of a flow $f_i$ are queued in its own dedicated queue.
A flow receives one service opportunity in each \emph{round} (see Algorithm~\ref{alg:WRR}).
Each flow is assigned a weight $w_i$ that sets the maximum number of packets that can be served in a single round.
Note that, since flow $f_i$ sends all $w_i$ packets consecutively, this may result in a considerable output burstiness for flows with higher weights.
The packet sizes are not fixed making the performance analysis of networks with variable packet sizes significantly more challenging.

Mathematically, we assume a service guarantee for the aggregate of all flows in form of a strict service curve or a work-conserving server which is typical in the literature \cite{PG93, BBLC18, TLBB21}.
Obtaining a per-flow service guarantee enables us to provide respective per-flow performance guarantees.

The second WRR variant we consider, interleaved weighted round robin (IWRR) (see Algorithm~\ref{alg:IWRR}), mitigates plain WRR's typical burstiness by supplementing the rounds with \emph{cycles}, while maintaining the same complexity.
Instead of receiving service for a burst of $w_i$ packets in a round, flow $f_i$ can only send one packet per cycle like any other flow which has not yet exhausted its weight allocation in the current round. After $w_i$ rounds flow $f_i$ has to wait until other flows $f_j$ with $w_j > w_i$ finish off the round.

\subsection{Network calculus analysis of (I)WRR}

It was shown in the literature that if $ \beta $ is assumed to be a constant-rate server using WRR, then rate-latency leftover service curves $ \beta_{R,T} $ can be derived.
Several publications determined the rate term $ R $ and latency term $ T $ under WRR, such as \cite{QLD09, GDR11, SLSF18} (for a detailed discussion, see \cite{TLBB21}).
Yet, rate-latency service curves do not yield tight results for WRR since each packet is only served with the traffic-agnostic residual rate \cite{TLBB21}.
Taking into account that a packet is transmitted at full link speed leads to a more precise model and consequently better service guarantees.
Such a leftover service curve, which is obviously not a rate-latency function anymore, is elegantly derived in \cite{BBLC18} for WRR (the improvement is depicted in Figure~\ref{fig:linear-staircase}).
We only state the two leftover service curves that do not require specific knowledge 
about possible packet sequences (so-called packet curves).

\begin{algorithm}[t]
	\caption{Interleaved Weighted Round Robin \cite{TLBB21}} \label{alg:IWRR}
	\hspace*{\algorithmicindent} \textbf{Input} Integer weights $ w_1, \dots, w_n $
	\begin{algorithmic}[1]
		\State $ w_{\max} \gets \max \left\{w_1, \dots, w_n\right\} $
		\While{True} \Comment{A round starts.}
		\For{$ C = 1 $ \textbf{to} $ w_{\max} $} \Comment{A cycle starts.}
		\For{$ i = 1 $ \textbf{to} $ n $}
		\If{\textbf{not} empty$ (i) $ \textbf{and} $ C \leq w_i $}
		\State \verb|send|(head($ i $));
		\State \verb|removeHead|($ i $);
		\EndIf
		\EndFor
		\EndFor \Comment{A cycle finishes.}
		\EndWhile \Comment{A round finishes.}
	\end{algorithmic}
\end{algorithm}

\begin{thm}[Strict Leftover Service Curves for WRR]
	\label{thm:WRR-SSC-BBLC18} 
	Assume $n$ flows arriving at a server
	performing weighted round robin (WRR) with weights $w_{1},\dots,w_{n}$.
	Let this server offer a strict service curve $\beta$ to these $n$
	flows. 
	We define $q_{i}\coloneqq w_{i}l_{i}^{\min}$ and $Q_{i}\coloneqq\sum_{j\neq i}w_{j}l_{j}^{\max}$.
	\begin{enumerate}[label=\arabic*.]
		\item Then, 
		\begin{equation*}
			\beta^i(t) \coloneqq \gamma_i'\left(\posPart{\beta(t)-Q_{i}}\right),
		\end{equation*}
		is a strict service curve for flow $f_{i}$,
		where we define
		\begin{align*}
			\gamma_i'(t) & \coloneqq \conv{\beta_{1,0}}{\nu_{q_{i},q_{i}+Q_{i}}}{t},
		\end{align*}
		the stair function
		\begin{equation}
			\nu_{h, P}(t)\coloneqq h\left\lceil \frac{t}{P}\right\rceil \quad\text{for }t\geq0,
			\label{eq:stair-function}
		\end{equation}
		and $\beta_{1,0}$ is a constant-rate function with slope 1.
		
		\item Moreover, 
		\begin{align*}
			\beta^{i}(t) & \coloneqq \frac{w_{i}l_{i}^{\min}}{w_{i}l_{i}^{\min}+\sum_{j\neq i}w_{j}l_{j}^{\max}}\posPart{\beta(t)-\sum_{j\neq i}w_{j}l_{j}^{\max}}\\
			&= \frac{q_{i}}{q_{i}+Q_{i}} \posPart{\beta(t)-Q_{i}}
		\end{align*}
		is a strict service curve for flow $f_{i}$.
		If $ \beta $ is a constant-rate server, then the residual service curve $ \beta^i $ is a rate-latency service curve.
	\end{enumerate}
\end{thm}
Note that for packets of constant size the delay bound of Theorem~\ref{thm:WRR-SSC-BBLC18} has been shown to be tight \cite{TLBB21}, i.e., a sample path attaining the delay bound was constructed.
To be precise, only the first part of the theorem, where each packet receives full link speed during transmission, is actually tight since the second part only provides the traffic-agnostic residual rate.
For a constant-rate server with rate $ C, $ this rate is equal to $ \frac{q_{i}}{q_{i}+Q_{i}} C. $

For IWRR, a leftover service curve has been derived in \cite{TLBB21}.
Not only do the authors show that exploiting the interleaving in the analysis can significantly improve delay bounds, they also prove that their bound is tight.
We only state the leftover service curve for IWRR.

\begin{thm}[Strict Leftover Service Curves for IWRR]
	\label{thm:IWRR-SSC-TLBB21} 
	Assume $n$ flows arriving at a server
	performing interleaved weighted round robin (IWRR) with weights $w_{1},\dots,w_{n}$.
	Let this server offer a strict service curve $\beta$ to these $n$
	flows. 
	Then, 
	\begin{equation*}
		\beta^i(t) \coloneqq \gamma_i''\left(\beta(t)\right),
	\end{equation*}
	is a strict service curve for flow $f_{i}$,
	where
	\begin{align}
		\gamma_i''(t)  \coloneqq & \conv{\beta_{1,0}}{U_i}{t}, \nonumber\\
		U_i (t)  \coloneqq & \sum_{k=0}^{w_i - 1} \nu_{l_i^{\min},L_\mathrm{tot}}\left(\posPart{t - \psi\left(k l_i^{\min}\right)}\right) \nonumber\\
		L_\mathrm{tot}  \coloneqq & q_i + Q_i \nonumber\\
		\psi_i(x)  \coloneqq & x + \sum_{j\neq i} \Psi_{ij} \left(\left\lfloor \frac{x}{l_{i}^{\min}}\right\rfloor \right) l_j^{\max} \nonumber\\
		\Psi_{ij}(p)  \coloneqq & \left\lfloor \frac{p}{w_{i}}\right\rfloor w_{j}+\posPart{w_{j}-w_{i}} \label{eq:phi-ij}\\
		&+\min\left\{ \left(p\mod w_{i}\right)+1,w_{j}\right\} \nonumber
	\end{align}
	and the stair function $ \nu_{h,P}(t) $  is defined in Equation~\eqref{eq:stair-function} as well as $q_{i} = w_{i}l_{i}^{\min}$ and $Q_{i} = \sum_{j\neq i}w_{j}l_{j}^{\max}$ again.
\end{thm}

\subsection{The case of constrained cross-traffic}
\label{subsec:case-cont-cross-traffic}

Both theorems, Theorem~\ref{thm:WRR-SSC-BBLC18} and \ref{thm:IWRR-SSC-TLBB21}, do not make any assumptions on the cross-traffic.
While this can be seen as a strength, in our work, we argue that by assuming cross-flows to be constrained, we open the door for improvement.
This can lead to significantly reduced delay bounds (see our numerical evaluation in Section~\ref{sec:numerical-evaluation}).
This reduction comes despite the fact that we do not closely model packet transmission at full link speed.
However, as we see in the next section, since we consider the maximum of shifted linear functions, we can often still obtain more than just the traffic-agnostic residual rate.
The central notion we employ to take knowledge on cross-flows into account is the bandwidth-sharing policy. 
We start off with its definition.

\begin{defn}[Bandwidth-Sharing Policy \cite{Bouillard21}] 
	A server has a \emph{bandwidth-sharing policy} if there exist positive weights $\phi_{i} > 0,i=1,\dots,n$ and nonnegative number $H_{ij} \geq 0, \, 1\leq i,j\leq n$ such that for a backlogged period $(s,t]$ of flow $f_{i}$ it holds that
	\begin{equation}
		\frac{D_{i}(s,t)}{\phi_{i}}\geq\frac{\posPart{D_{j}(s,t)-H_{ij}}}{\phi_{j}}, \quad \text{for all } j \neq i.
	\end{equation}
\end{defn}
Note that the bandwidth-sharing policy can be seen as a generalization of the resource allocation under GPS, where the $ H_{ij} $ would be 0.
We therefore interpret it as a penalty term that is a consequence of emulating the ideal fluid fair sharing of GPS by a real packet scheduler implementation.
For example, deficit round robin (DRR) has been show to be a bandwidth-sharing policy \cite{Bouillard21}.
For this class of schedulers, a profound result is given in the literature.
It was initially derived for GPS and then extended to schedulers which realize a bandwidth-sharing policy.

\begin{thm}[Strict Leftover Service Curves for Bandwidth-Sharing Policies]
	\label{thm:BWSP-SSC} 
	Assume $n$ flows arriving at a server
	with a bandwidth-sharing policy with positive weights $\phi_{i} > 0,i=1,\dots,n$ and a nonnegative penalty term $H_{ij} \geq 0, \, 1\leq i,j\leq n$.
	Let this server be a VCN that offers a convex $\beta$ to these $n$ flows. 
	Let $\mathcal{N} = \left\{ 1,\dots,n\right\} $ and assume that
	each flow $f_{i}$ is constrained by a concave arrival curve $\alpha_{i},i=1,\dots,n.$
	Then, 
	\begin{equation}
		\beta^{i}(t)= \max_{i\in M\subset\mathcal{N}}  \frac{\phi_{i}}{\sum_{k\in M}\phi_{k}}\posPart{\beta(t)-\sum_{k\notin M}\alpha_{k}(t)-\sum_{k\in M}H_{ik}}
	\end{equation}
	is a strict service curve for flow $f_{i}$.
\end{thm}
\begin{proof}
	See Theorem~1 in \cite{Bouillard21}.
\end{proof}
The number of possible sets $M$ to optimize over, of course, potentially becomes very large for a high number of flows: we have $ 2^{\left| \mathcal{N} \right| - 1 } $ subsets of $ \mathcal{N} $ containing $ i $. Yet, we point out that any selection of $M$ provides a strict leftover service curve. 
That means, we can come up with heuristics that avoid a combinatorial explosion by trading efficiency for the accuracy of the calculated bounds.
%Also note that in a homogeneous case, i.e., all functions and parameters being equal, this problem has just linear complexity. % Jens: may need some synchronization with the numerical evaluation part ...

\section{New Strict Service Curves for (Interleaved) Weighted Round Robin}
\label{sec:WRR-new-SC}

In this section, we show that WRR as well as IWRR are bandwidth-sharing policies.
This insight has the direct consequence of providing new leftover service curves which take into account arrival constraints on the cross-flows.

\begin{thm}[WRR is a Bandwidth-Sharing Policy]
	\label{thm:WRR-BWSP}
	Assume $n$ flows arriving at a server performing weighted round robin (WRR) with weights $w_{1},\dots,w_{n}$. Then, WRR is a bandwidth-sharing
	policy for flow $f_{i}$ with
	\begin{equation}
		\begin{aligned}\phi_{i}' &= w_{i} l_{i}^{\min} \eqqcolon q_i,\\
			\phi_{j}' & = w_{j} l_{j}^{\max} \eqqcolon q_j, \quad j \neq i
		\end{aligned}
		\label{eq:w-prime}
	\end{equation}
	and 
	\[
	H_{ij}' = w_{j}l_{j}^{\max}\IndicFun{i\neq j} = q_j\IndicFun{i\neq j},
	\]
	where $ \IndicFun{i\neq j} = 1 $ for $ i \neq j $ and $ 0 $, else.
	In other words, we have 
	\begin{equation}
		\frac{D_{i}(s,t)}{\phi_{i}'}\geq\frac{\posPart{D_{j}(s,t)-w_{j} l_{j}^{\max}\IndicFun{i\neq j}}}{\phi_{j}'}\label{ineq:WRR-bandwidth-sharing}
	\end{equation}
	for any $(s,t]$ such that flow $f_{i}$ is backlogged.
\end{thm}

\begin{proof}
	Following along the lines of \cite[pp. 201]{BBLC18}, we consider
	a backlogged period of $(s,t]$ of flow $f_{i}$ and let $p\in\mathbb{N}$
	denote the number of completed services of flow $f_{i}$ in the interval $(s,t].$
	Constructing the worst case, we have 
	\begin{equation}
		D_{i}(s,t)\geq pw_{i}l_{i}^{\min}.
		\label{ineq:WRR-Di-lower-bound}
	\end{equation}
	Moreover, this yields directly an upper bound for $p\in\mathbb{N}$:
	\begin{equation}
		p\leq\left\lfloor \frac{D_{i}(s,t)}{w_{i}l_{i}^{\min}}\right\rfloor .\label{ineq:WRR-p-upper-bound}
	\end{equation}
	On the other hand, in this interval, 
	\begin{equation}
		D_{j}(s,t)\leq(p+1)w_{j}l_{j}^{\max},\quad\forall j\neq i.\label{ineq:WWR-Dj-upper-bound}
	\end{equation}
	Summing the inequalities in \eqref{ineq:WRR-Di-lower-bound} and \eqref{ineq:WWR-Dj-upper-bound} yields
	\begin{align*}
		\frac{D_{i}(s,t)}{w_{i}}\overset{\eqref{ineq:WRR-Di-lower-bound},\eqref{ineq:WWR-Dj-upper-bound}}{\geq} & \frac{D_{j}(s,t)}{w_{j}}+pl_{i}^{\min}-(p+1)l_{j}^{\max}\\
		= & \frac{D_{j}(s,t)}{w_{j}}-p\left(l_{j}^{\max}-l_{i}^{\min}\right)-l_{j}^{\max}\\
		\overset{\eqref{ineq:WRR-p-upper-bound}}{\geq} & \frac{D_{j}(s,t)}{w_{j}}-\left\lfloor \frac{D_{i}(s,t)}{w_{i}l_{i}^{\min}}\right\rfloor \left(l_{j}^{\max}-l_{i}^{\min}\right)-l_{j}^{\max}\\
		\geq & \frac{D_{j}(s,t)}{w_{j}}-\frac{D_{i}(s,t)}{w_{i}}\cdot\frac{l_{j}^{\max}-l_{i}^{\min}}{l_{i}^{\min}}-l_{j}^{\max}.
	\end{align*}
	This is equivalent to
	%	\[
	%	\frac{D_{i}(s,t)}{w_{i}}\underbrace{\left(1+\frac{l_{j}^{\max}-l_{i}^{\min}}{l_{i}^{\min}}\right)}_{=l_{j}^{\max}/l_{i}^{\min}}\geq\frac{D_{j}(s,t)}{w_{j}}-l_{j}^{\max}
	%	\]
	%	and thus,
	\[
	\frac{D_{i}(s,t)}{w_{i} l_{i}^{\min}}\geq\frac{1}{l_{j}^{\max}}\left(\frac{D_{j}(s,t)}{w_{j}}-l_{j}^{\max}\right)=\frac{D_{j}(s,t) - w_{j} l_{j}^{\max}}{w_{j} l_{j}^{\max}}.
	\]
	Inserting $\phi_{i}',\phi_{j}'$ as in Equation~\eqref{eq:w-prime} and
	using that $D_{i}(s,t)\geq0$ for all $0\leq s\leq t$ finishes the
	proof.
\end{proof}

We interpret the assignment for $ \phi_i', i = 1, \dots,n $ in Equation~\eqref{eq:w-prime} as the weight corrected to the amount of data per round in the worst case from the perspective of the flow of interest.
%Paul: new stuff

\begin{cor}[Strict Leftover Service Curve for WRR]
	\label{cor:WRR-Vlad} Assume $n$ flows arriving at a server
	performing weighted round robin (WRR) with weights $w_{1},\dots,w_{n}$.
	Let this server be a VCN that offers a convex $\beta$ to these $n$ flows. 
	Let $\mathcal{N}\coloneqq\left\{ 1,\dots,n\right\} $ and assume that each flow $f_{i}$ is constrained by a concave arrival curve $\alpha_{i},i=1,\dots,n.$
	We define
	\begin{align*}
		%		q_{i} & \coloneqq w_{i}l_{i}^{\min},\\
		Q_{M}' & \coloneqq\sum_{k\in M\setminus\left\{ i\right\} }w_{k}l_{k}^{\max}.
	\end{align*}
	Then, 
	\begin{equation}
		\begin{aligned}
			\beta^{i}(t) & =\max_{i\in M\subset\mathcal{N}}  \frac{q_{i}}{q_{i}+Q_{M}'} \posPart{\beta(t)-\sum_{k\notin M}\alpha_{k}(t)-Q_{M}'}
		\end{aligned}
		\label{eq:WRR-beta-i-M}
	\end{equation}
	is a strict service curve for flow $f_{i}$.
	In the following, we call the leftover service curve $ \beta^{i} $ \emph{WRR M}.
\end{cor}

\begin{proof}
	The corollary is a consequence of WRR being a bandwidth sharing policy (Theorem~\ref{thm:WRR-BWSP}) 
	($\phi_{i}',i=1,\dots,n$ and $ H_{ik} $ constructed above) 
	together with Theorem~\ref{thm:BWSP-SSC} which gives a strict service curve for these policies.
\end{proof}
One can easily show that if $ \beta $ is a constant-rate server, then $ \beta^i $ is a rate-latency service curve for any $ i \in M \subset \mathcal{N} $.  

Furthermore, note that Corollary~\ref{cor:WRR-Vlad} can be relaxed in the sense that we actually do not need all cross-flows to be constrained by arrival curves.
For these flows $ f_k $ we can simply define $ \alpha_{k}(t) = \infty $ 
%for $ t > 0 $ (0 otherwise) 
and they are going to be an element of $ M $, i.e., they do not increase the sum $ \sum_{k\notin M} \alpha_{k}(t), $ when maximizing over $ M $ in Equation~\eqref{eq:WRR-beta-i-M}.
The same applies when the long-term arrival rate of $ f_k $, $ \lim_{t \to \infty} \frac{\alpha_k(t)}{t}, $ is actually larger than the long-term server rate of $ \beta. $
In other words, similar to the state of the art, we are not limited to stable systems.

Let us explain the actual gain compared to the state of the art under token-bucket constrained arrivals and a constant-rate server with rate $ C $.
Compared to Theorem~\ref{thm:WRR-SSC-BBLC18}.2, we have now the possibility to change the latency for $ M \subsetneq \mathcal{N} $ ($ \sum_{k\notin M}\gamma_{r_k, b_k}(t) $ is decreasing in $ M $, while $ Q_M' $ is increasing) and at the same time change the traffic-agnostic residual rate to $ \frac{q_{i}}{q_{i}+Q_{M}'} \left(C - \sum_{k \notin M} r_k\right). $
As an obvious consequence, the improvement impact is increased the more the cross-traffic can be constrained.
Additionally, Corollary~\ref{cor:WRR-Vlad} directly recovers Theorem~\ref{thm:WRR-SSC-BBLC18}.2.
However, the relation to the stair function in Theorem~\ref{thm:WRR-SSC-BBLC18}.1. is not as clear.
In our numerical evaluation (Section~\ref{sec:numerical-evaluation}), we show that our leftover service curve can outperform the one in Theorem~\ref{thm:WRR-SSC-BBLC18}.1. and lead to better delay bounds.
%In our numerical evaluation (Section~\ref{sec:numerical-evaluation}), we show that both leftover service curves can outperform each other and lead to better delay bounds.

Next, we continue with the analysis of interleaved weighted round robin.
It mitigates burstiness by employing a simple trick, namely by introducing cycles that prevents flows from sending data consecutively.
Yet, these cycles make the mathematical analysis more difficult.
In this section, we show, based on some insights presented in \cite{TLBB21}, that IWRR is also a bandwidth-sharing policy.
To be precise, we show it for two different penalty terms.
While one is obtained by exploiting peculiarities of IWRR, the other is the same as for WRR.
Therefore, our new strict service curve for IWRR can only be equal or better than the one we obtained for WRR.

\begin{thm}[IWRR is a Bandwidth-Sharing Policy]
	\label{thm:IWRR-BWSP}
	Assume $n$ flows arriving at a server performing interleaved
	weighted round robin (IWRR) with weights $w_{1},\dots,w_{n}$. 
	Then,
\end{thm}

\begin{enumerate}[label=\arabic*.]
	\item IWRR is a bandwidth-sharing policy for flow $f_{i}$,
	where 
	\begin{equation}
		\begin{aligned}\phi_{i}'' & =w_{i}l_{i}^{\min}=q_{i},\\
			\phi_{j}'' & =\left(w_{j}+w_{i}\right)l_{j}^{\max}=q_{j}+w_{i}l_{j}^{\max}
		\end{aligned}
		\label{eq:w-double-prime}
	\end{equation}
	and 
	\[
	H_{ij}'' = \left(\posPart{w_{j}-w_{i}}+1\right)l_{j}^{\max}\IndicFun{i\neq j}.
	\]
	In other words, we have 
	\begin{equation}
		\frac{D_{i}(s,t)}{\phi_{i}''}\geq\frac{\posPart{D_{j}(s,t)-\left(\posPart{w_{j}-w_{i}}+1\right)l_{j}^{\max}\IndicFun{i\neq j}}}{\phi_{j}''}\label{ineq:IWRR-bandwidth-sharing}
	\end{equation}
	for any $(s,t]$ such that flow $f_{i}$ is backlogged. 
	\item IWRR is bandwidth-sharing policy with the $\phi_{i}',\phi_{j}',H_{ij}'$
	from WRR.
\end{enumerate}

%For the sake of conciseness, we provide a full proof of the first part and only a sketch of the second.

\begin{proof}[Proof of 1]
	Let $(s,t]$ be a backlogged period of flow $f_{i}$. 
	Again, let $p\in\mathbb{N}$ denote the number of completed services of flow $f_{i}$ in the interval $(s,t].$
	By construction, it holds that
	\begin{equation}
		D_{i}(s,t)\geq pl_{i}^{\min}.\label{ineq:IWRR-Di-lower-bound}
	\end{equation}
	Note that, in comparison to WRR, we need to adjust the inequality since we have cycles in between rounds.
	As a direct consequence of Equation~\eqref{ineq:IWRR-Di-lower-bound}, we can upper bound $p\in\mathbb{N}$ by
	\begin{equation}
		p\leq\left\lfloor \frac{D_{i}(s,t)}{l_{i}^{\min}}\right\rfloor .\label{ineq:IWRR-p-upper-bound}
	\end{equation}
	Moreover, by combining Lemma~4 and Lemma 6 in \cite{TLBB21}, we receive
	\begin{equation}
		D_{j}\left(s,t\right)\leq\Psi_{ij}(p)l_{j}^{\max},\quad\forall j\neq i,\label{ineq:IWRR-Dj-upper-bound}
	\end{equation}
	where is $ \Psi_{ij}(p) $ is defined in Equation~\eqref{eq:phi-ij}.
	Summing the inequalities in \eqref{ineq:IWRR-Di-lower-bound} and \eqref{ineq:IWRR-Dj-upper-bound} yields
	\begin{align}
		& \frac{D_{i}(s,t)}{w_{i}}\nonumber \\
		\geq & \frac{D_{j}(s,t)}{w_{j}}-\frac{p}{w_{i}}\left(l_{j}^{\max}-l_{i}^{\min}\right)\nonumber \\
		& -\frac{\left(\posPart{w_{j}-w_{i}}+\min\left\{ \left(p\mod w_{i}\right)+1,w_{j}\right\} \right)l_{j}^{\max}}{w_{j}}\nonumber \\
		\overset{\eqref{ineq:IWRR-p-upper-bound}}{\geq} & \frac{D_{j}(s,t)}{w_{j}}-\frac{\left\lfloor \frac{D_{i}(s,t)}{l_{i}^{\min}}\right\rfloor }{w_{i}}\left(l_{j}^{\max}-l_{i}^{\min}\right) \label{ineq:aux-ineq-part2}\\
		& -\frac{\left(\posPart{w_{j}-w_{i}}+\min\left\{ \left\lfloor \frac{D_{i}(s,t)}{l_{i}^{\min}}\right\rfloor \! \mod w_{i}+1,w_{j}\right\} \right)l_{j}^{\max}}{w_{j}}\nonumber \\
		\geq & \frac{D_{j}(s,t)}{w_{j}}-\frac{\left\lfloor \frac{D_{i}(s,t)}{l_{i}^{\min}}\right\rfloor }{w_{i}}\left(l_{j}^{\max}-l_{i}^{\min}\right)\label{ineq:p-bounded-Di-applied}\\
		& -\frac{\left(\posPart{w_{j}-w_{i}}+ \left\lfloor \frac{D_{i}(s,t)}{l_{i}^{\min}}\right\rfloor +1\right)l_{j}^{\max}}{w_{j}}\nonumber \\
		\geq & \frac{D_{j}(s,t)}{w_{j}}-\frac{D_{i}(s,t)}{w_{i}}\cdot\frac{1}{l_{i}^{\min}}\left(l_{j}^{\max}-l_{i}^{\min}\right)\nonumber \\
		& -\frac{D_{i}(s,t)}{w_{j}}\cdot\frac{l_{j}^{\max}}{l_{i}^{\min}}-\frac{\left(\posPart{w_{j}-w_{i}}+1\right)l_{j}^{\max}}{w_{j}}.\nonumber 
	\end{align}
	Note that in Equation~\eqref{ineq:p-bounded-Di-applied}, we upper bounded
	$ \min\left\{ \left(\left\lfloor \frac{D_{i}(s,t)}{l_{i}^{\min}}\right\rfloor \! \mod w_{i}\right) + 1, w_{j}\right\} $
	by $ \left\lfloor \frac{D_{i}(s,t)}{l_{i}^{\min}}\right\rfloor $. 
	Above inequality is equivalent to
	%	\[
	%	\underbrace{\left(1+\cdot\frac{1}{l_{i}^{\min}}\left(l_{j}^{\max}-l_{i}^{\min}\right)+\frac{w_{i}}{w_{j}}\cdot\frac{l_{j}^{\max}}{l_{i}^{\min}}\right)}_{=\left(1+\frac{w_{i}}{w_{j}}\right)\frac{l_{j}^{\max}}{l_{i}^{\min}}}\frac{D_{i}(s,t)}{w_{i}}\geq\frac{D_{j}(s,t)}{w_{j}}-\frac{\left(\posPart{w_{j}-w_{i}}+1\right)l_{j}^{\max}}{w_{j}}
	%	\]
	%	which in turn yields
	\begin{align*}
		& \frac{D_{i}(s,t)}{w_{i}l_{i}^{\min}}\\
		\geq & \frac{w_{j}}{w_{j}+w_{i}}\cdot\frac{D_{j}(s,t)}{w_{j}l_{j}^{\max}}-\frac{w_{j}}{w_{j}+w_{i}}\cdot \frac{\left(\posPart{w_{j}-w_{i}}+1\right)l_{j}^{\max}}{w_{j}l_{j}^{\max}}.
	\end{align*}
	Replacing $w_{i}$ and $w_{j}$ with $\phi_{i}''$ and $\phi_{j}''$, respectively, from Equation~\eqref{eq:w-double-prime} yields
	\[
	\frac{D_{i}(s,t)}{\phi_{i}''}\geq\frac{D_{j}(s,t) - \left(\posPart{w_{j}-w_{i}}+1\right)l_{j}^{\max}}{\phi_{j}''}.
	\]
	In the final step, we use again that $ D_i(s,t) \geq 0. $
\end{proof}

%\begin{proof}[Sketch of 2]
%	Starting in Equation~\eqref{ineq:aux-ineq-part2}, we could also continue to upper bound 
%	$ \min\left\{ \left(\left\lfloor \frac{D_{i}(s,t)}{l_{i}^{\min}}\right\rfloor \! \mod w_{i}\right) + 1, w_{j}\right\} $ by $ \min\left\{ w_{i},w_{j}\right\}. $
%	Applying again Equation~\eqref{ineq:IWRR-p-upper-bound} and using that
%	\[
%	\posPart{w_{j}-w_{i}}+\min\left\{ w_{i},w_{j}\right\} = w_{j}
%	\]
%	gives the result.
%\end{proof}

\begin{proof}[Proof of 2]
	Starting in Equation~\ref{ineq:aux-ineq-part2}, we could also
	continue as follows:
	\begin{align*}
		& \frac{D_{i}(s,t)}{w_{i}}\\
		\overset{\eqref{ineq:aux-ineq-part2}}{\geq} & \frac{D_{j}(s,t)}{w_{j}}-\frac{\left\lfloor \frac{D_{i}(s,t)}{l_{i}^{\min}}\right\rfloor }{w_{i}}\left(l_{j}^{\max}-l_{i}^{\min}\right)\\
		& -\frac{\left(\posPart{w_{j}-w_{i}}+\min\left\{ \left(\left\lfloor \frac{D_{i}(s,t)}{l_{i}^{\min}}\right\rfloor \mod w_{i}\right)+1,w_{j}\right\} \right)l_{j}^{\max}}{w_{j}}\\
		\geq & \frac{D_{j}(s,t)}{w_{j}}-\frac{\left\lfloor \frac{D_{i}(s,t)}{l_{i}^{\min}}\right\rfloor }{w_{i}}\left(l_{j}^{\max}-l_{i}^{\min}\right)\\
		& -\frac{\left(\posPart{w_{j}-w_{i}}+\min\left\{ w_{i},w_{j}\right\} \right)l_{j}^{\max}}{w_{j}}\\
		\geq & \frac{D_{j}(s,t)}{w_{j}}-\frac{D_{i}(s,t)}{w_{i}}\left(\frac{l_{j}^{\max}}{l_{i}^{\min}}-1\right)\\
		& -\frac{\left(\posPart{w_{j}-w_{i}}+\min\left\{ w_{i},w_{j}\right\} \right)l_{j}^{\max}}{w_{j}}.
	\end{align*}
	Note that it holds
	\begin{align*}
		\posPart{w_{j}-w_{i}}+\min\left\{ w_{i},w_{j}\right\}  & =\begin{cases}
			w_{j}, & \text{if }w_{i}\leq w_{j},\\
			w_{j}, & \text{otherwise}
		\end{cases}\\
		& =w_{j}.
	\end{align*}
	Above inequality is therefore equivalent to
	\[
	\frac{D_{i}(s,t)}{w_{i}l_{i}^{\min}}\geq\frac{D_{j}(s,t)-w_{j}l_{i}^{\min}}{w_{j}l_{j}^{\max}}.
	\]
	Replacing $w_{i}$ and $w_{j}$ with $\phi_{i}'$ and $\phi_{j}'$
	and using that $D_{1}(s,t)\geq0$ yields the result.
\end{proof}

We can now formulate the leftover service curve.

\begin{figure*}[t]
	\centering
	\subfloat[\centering Service curves (utilization$~=~0.6$) \label{fig:sc-comparison}]{
		\includegraphics[height = 0.2\textheight]{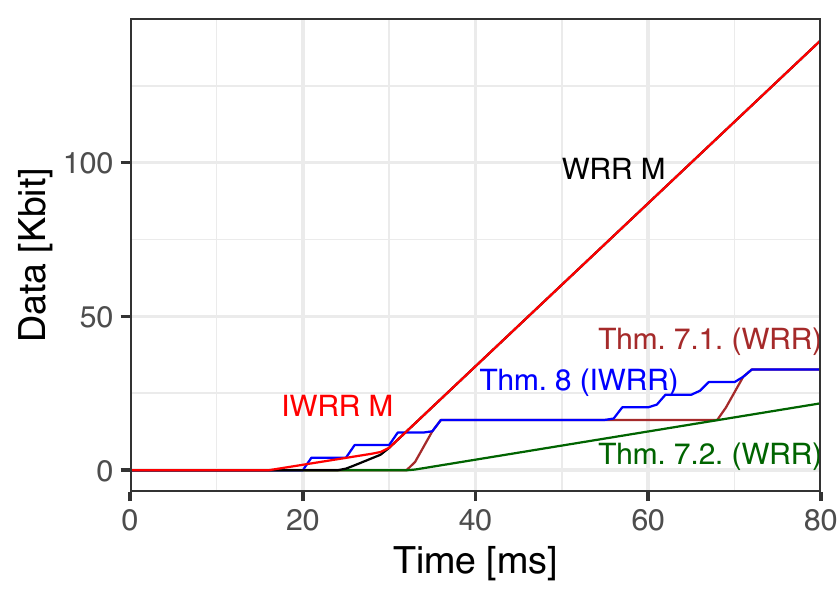}
	}
	\hspace{6mm}
	\subfloat[\centering Delay bounds \label{fig:delay-bound-comparison}]{
		\includegraphics[height = 0.2\textheight]{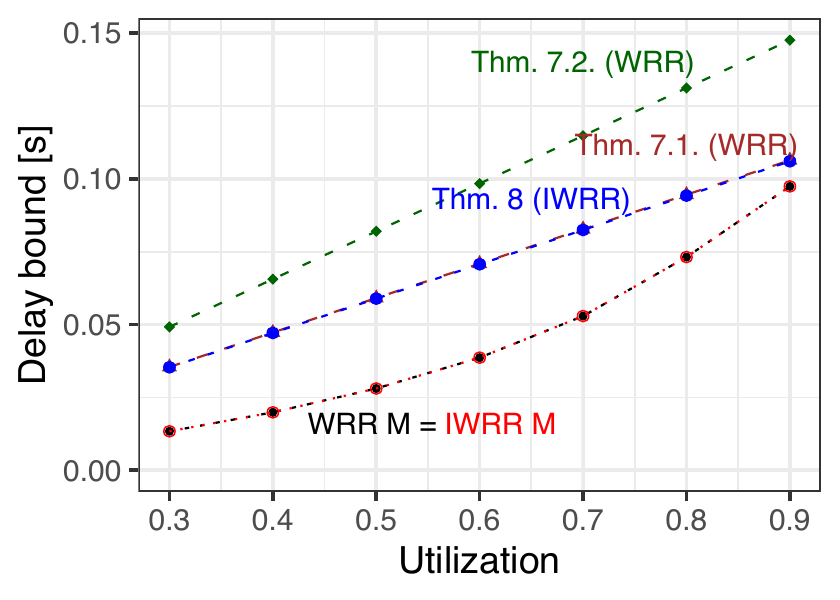}
	}
	\caption{\centering Parameter set: weights$~=~\{4, 6, 7, 10\}, $ $ l^{\min}~=~\{4096, 3072, 4608, 3072\} $ \si{bits}, $ l^{\max}~=~\{8704, 5632, 6656, 8192\} $ \si{bits}, burst sizes$~=~\{30208, 19968, 24576, 27648\}$ \si{bits}, arrival rates$~=~\{0.65, 0.85, 0.95, 0.55\} $ \si{Mbit \per s}.}
	%	3 \cdot l^{\max} + l^{\min}
\end{figure*}

\begin{cor}[Strict Leftover Service Curve for IWRR]
	\label{cor:IWRR-Vlad} Assume $n$ flows arriving at a server
	performing interleaved weighted round robin (IWRR) with weights $w_{1},\dots,w_{n}$.
	Let this server be a VCN that offers a convex $ \beta $ to these $n$ flows. 
	Let $\mathcal{N}\coloneqq\left\{ 1,\dots,n\right\} $ and assume that
	each flow $f_{i}$ is constrained by a concave arrival curve $\alpha_{i},i=1,\dots,n.$
	We define
	\[
	Q_{M}''\coloneqq\sum_{k\in M\setminus\left\{ i\right\} }w_{k}\left(\left(1+\frac{w_{i}}{w_{k}}\right)\cdot l_{k}^{\max}\right).
	\]
	and
	\[
	H_{ik}'' = \left(\posPart{w_{k}-w_{i}}+1\right)l_{k}^{\max} \IndicFun{i\neq k}.
	\]
	Then, 
	\begin{equation}
		\beta^{i}(t) =\max_{i\in M\subset\mathcal{N}}  \left\lbrace \max \beta_{1, M}^i, \beta_{2, M}^i \right\rbrace 
	\end{equation}
	is a strict service curve for flow $f_{i}$, where 
	\begin{equation}
		\begin{aligned}
			&\beta_{1, M}^i(t)\\
			=& \frac{q_{i}}{q_{i}+Q_{M}''}\posPart{\beta(t)-\sum_{k\notin M}\alpha_{k}(t)-\sum_{k\in M\setminus\left\{ i\right\}} H_{ik}''}
		\end{aligned}
	\end{equation}
	and $ \beta_{2, M}^i $ is the strict leftover service curve for WRR (Equation~\eqref{eq:WRR-beta-i-M}).
	In the following, we call the leftover service curve $ \beta^{i} $ \emph{IWRR M}.
\end{cor}
\begin{proof}
	Similar to the WRR case, Theorem~\ref{thm:IWRR-BWSP} proves that IWRR is a bandwidth sharing policy for two different penalty terms.
	Theorem~\ref{thm:BWSP-SSC} then gives us two different strict leftover service curve.
	At last, we use that the maximum of strict service curves is again a strict service curve \cite[p. 109]{BBLC18}.
	Note that this property is actually already used in the derivation of Theorem~\ref{thm:BWSP-SSC}.  
\end{proof}

Again, if $ \beta $ is a constant-rate server, then $ \beta_{1, M} $ and $ \beta_{2, M} $ are rate-latency service curves for any choice of $ i \in M \subset \mathcal{N} $.

\section{Numerical Evaluation}
\label{sec:numerical-evaluation}

In this section, we numerically compare our newly obtained leftover service curves under constrained cross-traffic, WRR M (Corollary~\ref{cor:WRR-Vlad}) and IWRR M (Corollary~\ref{cor:IWRR-Vlad}), to the state of the art in Theorem~\ref{thm:WRR-SSC-BBLC18} (WRR) and Theorem~\ref{thm:IWRR-SSC-TLBB21} (IWRR).

We assume all flows $ f_i $ to be constrained by a token bucket arrival curve $ \gamma_{r_i,b_i} $, $ i = 1, \dots, n $.
For the service, we always assume a constant server rate $ C > 0 $ for the aggregate of flows.

%Short overview?
In our numerical experiments, we first compare our results against the state of the art in a literature example and then evaluate the impact of factors such as the burst sizes of cross-flows, the (maximum and minimum) packet sizes, and the number of cross-flows. Last, we also take a look at scenarios with larger numbers of flows in which we need to search for the optimal set $M$ in our leftover services curves heuristically.

\subsection{Comparison to state of the art}

First, let us consider the example presented in \cite{TLBB21}.
We start off with a direct comparison of the service curves.
The results are given in Figure~\ref{fig:sc-comparison}.

As expected, the stair function from Theorem~\ref{thm:WRR-SSC-BBLC18}.1 always provides a larger service curve than the rate-latency variant in Theorem~\ref{thm:WRR-SSC-BBLC18}.2.
Moreover, we can also see the positive effect of the cycles within rounds for the IWRR service curve (Theorem~\ref{thm:IWRR-SSC-TLBB21}).
%Both new leftover service curves, WRR M and IWRR M, are significantly
The first new leftover service curve, WRR M, is significantly
larger due to the larger traffic-aware residual rate, except for the very start due to a larger latency period (recalling the illustrative Figure~\ref{fig:linear-staircase} and the corresponding discussion in the introduction).
%Note that, while IWRR M has a smaller latency period, the curve is then dominated by the
IWRR M, on the other hand, has a smaller latency period, yet the curve is then dominated by the larger rate $ \frac{q_{i}}{q_{i}+Q_{M}'} \left(C - \sum_{k \notin M} r_k\right) $ of WRR M.
This observation is consistent with our expectation, since we invoked more inequalities in the proof of IWRR being a bandwidth-sharing policy.
Therefore, it is more difficult for the IWRR M service curve to benefit from the interleaving.

Next, we continue by comparing the impact of the service curves on the delay bounds (calculated using Theorem~\ref{thm:delay-bound}). 
Note that, for Theorem~\ref{thm:WRR-SSC-BBLC18}.2. and our leftover service curves in Corollary~\ref{cor:WRR-Vlad} and \ref{cor:IWRR-Vlad}, we receive rate-latency functions and we can therefore apply closed-form solutions for the delay bound \cite[p. 24]{LBT01}.
For the stair functions, on the other hand, calculation is more complex \cite{BT08}.
The results are depicted in Figure~\ref{fig:delay-bound-comparison}.

As expected, the stair function for WRR leads to better delay bounds than the result with the traffic-agnostic residual rate, while the interleaving reduces delay bounds even further.
However, most importantly, our new service curves, taking cross-traffic into account, lead to significantly smaller delay bounds compared to all other curves.
The decreasing gain over the state of the art for higher utilizations is expected, since for high utilizations all flows tend to be backlogged most of the time which forces the traffic-aware residual rate to get ever closer to the traffic-agnostic one.

\subsection{Impact of burst sizes}

It is clear that our new leftover service curves depend on the burst sizes of the cross-flows, while the state of the art is oblivious to it and is only affected by the maximum packet sizes as well as the weights of the cross-traffic.
Therefore, in this numerical experiment we investigate this impact.

We consider three different classes of cross-traffic: low-, mid-, and high-burstiness flows in addition to the flow of interest (foi).
We distribute 9 cross-flows over these 3 classes such that we start off with a scenario of cross-traffic with low burstiness and then gradually turn it into a scenario with high burstiness.
Specifically, we denote by $ \left(n_{\mathrm{low}}, n_{\mathrm{mid}}, n_{\mathrm{high}}\right) $ the number of flows of each ``burst class''. 
The results are given in Figure~\ref{fig:burst-class-analysis}.

As expected, we observe that for $ \left(7, 1, 1\right) $, the case with the smallest cross-flow burstiness, the gain of using this information is the largest because the latency increase in our leftover service curves is small.
In this case, the delay bounds are reduced by more than half compared to the state of the art.
Increasing the burstiness, of course, reduces this advantage.
However, it stays below even in the scenario with the highest burstiness.

\begin{figure}[t]
	\centering
	\includegraphics[height = 0.2\textheight]{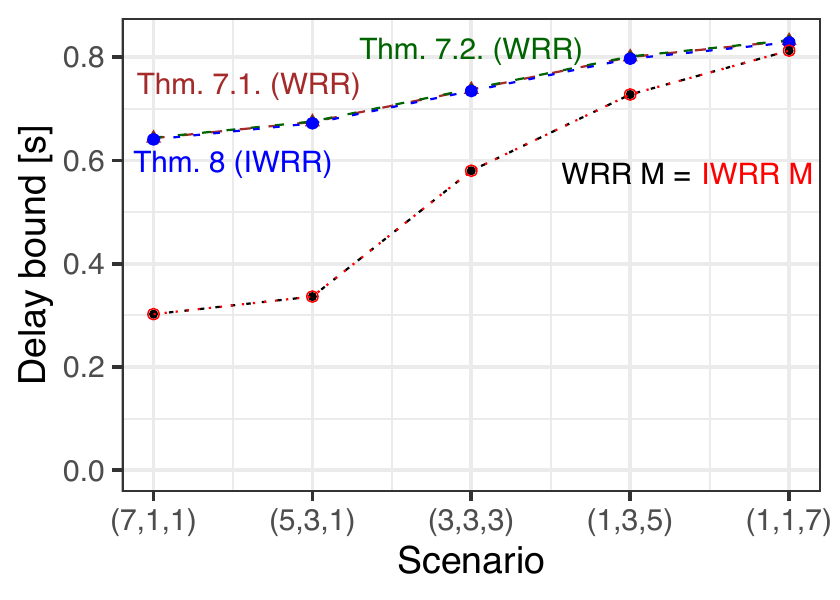}
	\caption{For the parameter description of traffic classes, we use the notation $ \{\mathrm{low}, \mathrm{mid}, \mathrm{high}, \mathrm{foi}\} $: weights$~=~\{4, 5, 6, 5\}, $ burst sizes$~=~\{70, 700, 7000, 3000\}~$\si{Kbit}. We use the packet lengths $ l_i^{\min}~=~576 $ \si{bytes}, $ l_i^{\max}~=~1500 $ \si{bytes} and arrival rates$~=~7 $ \si{Mb \per s} for each flow and a server utilization$~=~0.7 $. \label{fig:burst-class-analysis}}
\end{figure}

\subsection{Impact of packet sizes}

The accuracy of all performance bounds depends on the variability of the packet sizes.
In this numerical experiment, we examine how the ratio between maximum and minimum packet size of the flows impacts the delay bounds.
To that end, we define the packet size range (PSR) as
\[
\text{PSR} \coloneqq \frac{\max_{i=1, \dots, n} l_i^{\max}}{\min_{i=1, \dots, n} l_i^{\min}}.
\]
In the experiment, we start off with packets of equal size, that is, a PSR of 1, and then, by decreasing minimum packet sizes, increase the PSR. 
The results are displayed in Figure~\ref{fig:packet-size-range}.
Clearly, the delay bounds increase when the PSR increases.
However, we observe that our analysis is affected less and the gain of our results over the state of the art increases.
The likely reason is that our leftover service curves can better compensate for the higher packet size variability using the additional degree of freedom from the selection of flows that are assigned to the set $M$.

\begin{figure}[t]
	\centering
	\includegraphics[height = 0.2\textheight]{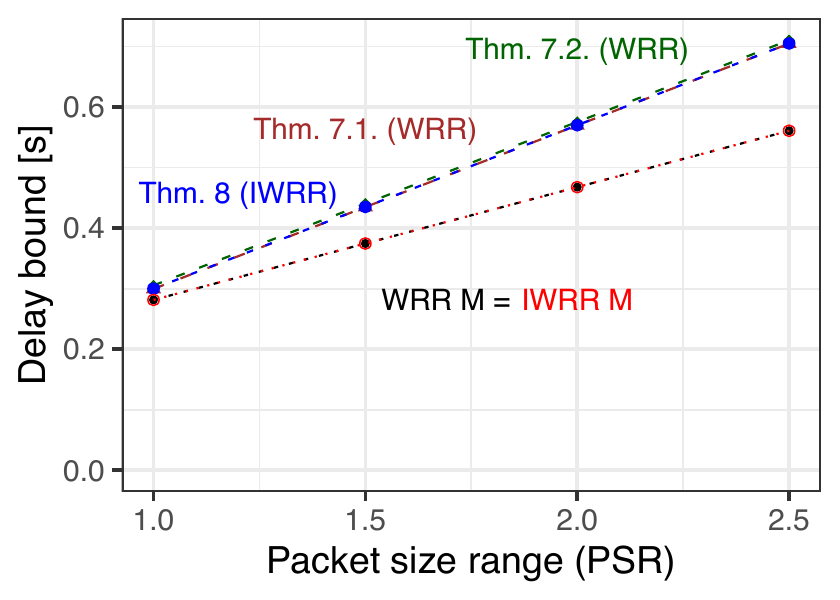}
	\caption{Impact of the packet size range on the delay bounds. We choose the same parameters setting as in the previous experiment, except for the minimum packet size $ l_i^{\min} , i = 1, \dots,n$.}
	\label{fig:packet-size-range}
\end{figure}

\subsection{Impact of number of cross-flows}

In the next experiment, we investigate the impact of the number of cross-flows.
Specifically, we consider again the three ``burst classes'' with the same number of flows in each class, yet now increasing the number of flows per class.
If $ k $ is the number of flows per class, clearly, we have $ 3k $ cross-flows in total. 
The results are depicted in Figure~\ref{fig:number-of-flows}.
We observe that our analysis improves on the state of the art by more than $20$\% across the different numbers of cross-flows.

\begin{figure}[t]
	\centering
	\includegraphics[height = 0.2\textheight]{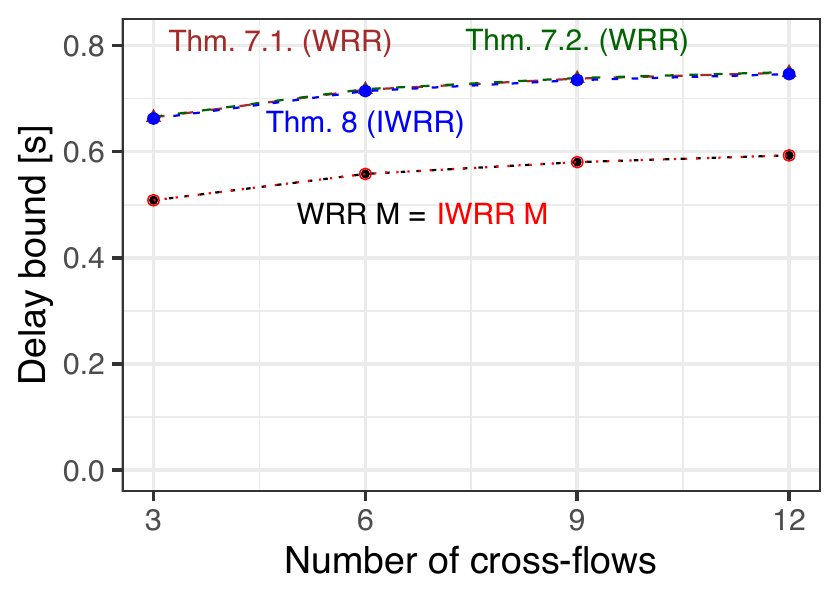}
	\caption{The parameter setting is the same as in Figure~\ref{fig:burst-class-analysis}, except the number of cross-flows per class. \label{fig:number-of-flows}}
\end{figure}

\subsection{Dealing with larger number of flows}

\begin{table}[t]
	\centering
	\begin{tabular}{|c|c|c|}
		\hline
		Total number of flows & Delay bound heuristic & Delay bound Theorem~\ref{thm:IWRR-SSC-TLBB21} \\ \hline
		13 & 0.59 &0.75 \\ \hline
		49 & 0.62 & 0.78 \\ \hline
		100 & 0.63 & 0.78 \\ \hline
		499 & 0.63 & 0.79 \\ \hline
		1000 & 0.64 & 0.79 \\ \hline
	\end{tabular}
	\caption{\centering Delay bound comparison between the heuristic and the state of the art under the parameters of Figure~\ref{fig:burst-class-analysis}.}
	\label{tab:heuristic-comparison}
\end{table}

In the previous experiment, we investigated the impact of the number of cross-flows on the delay bounds.  The total number of flows was kept relatively moderate. In fact, as briefly discussed in Section~\ref{subsec:case-cont-cross-traffic}, if we want to deal with larger number of flows, we need to take into account that our leftover service curves 
rely on maximizing all possibly subsets $ M $ such that $ i \in M \subset \mathcal{N}, $ i.e., $ 2^{\left|\mathcal{N} \right|  - 1} $ combinations.
Therefore, for a large number of flows, finding the optimal $ M $ to minimize the delay bound becomes computationally prohibitive. However, we do not have to evaluate all possible subsets but can actually apply a search heuristic to this combinatorial problem. Here, we briefly propose a very efficient and simple one:
\begin{itemize}
	\item Let $ d_{M} $ denote the delay bound of WRR M for $ i \! \in \! M \!\subset \! \mathcal{N} $.
	\item Set $ d_{\mathrm{opt}} := \infty $ and $ M_{\mathrm{opt}} := \{i\}$.
	\item Sort all cross-flows in a descending order by burst size, with $j$ being the sorted flow index.
	\item For $ j $ from 1 to $ |\mathcal{N}| -1 $\\ 
	\vspace{1mm} \quad calculate $ d_{\mathrm{new}}~=~\min \left\{d_{\mathrm{opt}},  d_{M_{\mathrm{opt}}\cup \{j\}} \right\} $\\
	\vspace{1mm} \quad if $ d_{\mathrm{new}} < d_{\mathrm{opt}} $ \\
	\vspace{1mm} \qquad then set $ d_{\mathrm{opt}} \coloneqq d_{\mathrm{new}} $ and $ M_{\mathrm{opt}} := M_{\mathrm{opt}}\cup \{j\} $.
\end{itemize}
We compare this heuristic WRR M  with the state of the art with the smallest delay bounds, Theorem~\ref{thm:IWRR-SSC-TLBB21} (IWRR), for large numbers of flows, again from the three classes as in previous experiments.

The results are given in Table~\ref{tab:heuristic-comparison}.
We see that our heuristic yields significantly smaller delay bounds across the different numbers of flows. 
We measured a runtime of $ 39.3 \si{s} $ to find the optimal $ M, $ while the heuristic took only $ 4 \cdot 10^{-4} \si{s}. $
Even for the largest scenario with 1000 flows, the heuristic completed the search in less than $29.9\si{s}.$

% 100 flows scenario:
% TLBB22 delay bound: 0.778, it took 1 h
% WRR M Heu: 0.6322695, it took 10s

\section{Conclusion}
\label{sec:conclusion}

In this paper, we have improved performance bounds of (interleaved) weighted round robin under the assumption of constrained cross-traffic.
To that end, we showed that both discussed WRR variants are bandwidth-sharing policies.
For WRR, we gained more insights by refining the flow weights with the respective worst-case packet lengths.
Consequently, we exploited this property to derive new strict leftover service curves for (I)WRR.
In a numerical evaluation, we observed that the improvement is not only substantial, but persistent across different experiments investigating the impact of several factors.

For future work, the extension to other round-robin schedulers and proving the bandwidth-sharing property for stair functions are interesting research challenges. Furthermore, motivated by the promising results, it will be very interesting to investigate systematically the heuristic optimization of the new leftover service curves for round-robin schedulers.

\bibliographystyle{alpha}
\bibliography{biblio.bib}

\end{document}